\documentclass[a4paper]{article}
\usepackage[all]{xy}
\usepackage{amsmath,amssymb,amsthm,mathrsfs,latexsym,a4wide}
\newtheorem{thm}{Theorem}[section]
\newtheorem{definition}[thm]{Definition}
\newtheorem{theorem}[thm]{Theorem}
\newtheorem{lemma}[thm]{Lemma}
\newtheorem{corollary}[thm]{Corollary}

\begin{document}

\title{Central Configurations Formed By Two Twisted Regular Polygons}
\author{ \small\sc Xiang Yu\footnote{Email:xiang.zhiy@gmail.com} \small{and}
\small\sc Shiqing Zhang\footnote{Email:zhangshiqing@msn.com} \\
 \small \it Department of
Mathematics, Sichuan University,
 \small\it Chengdu 610064,  China}
\date{}
\maketitle
\begin{abstract}In this paper, we study the necessary conditions and sufficient conditions for the central configurations formed by two twisted regular polygons (one N-regular polygon and one L-regular polygon ). We wish to extend the results of the symmetrical  central configurations formed by two twisted N-regular polygons, however, it  will be  proved that there are not more central configurations in a more general setting than the central configurations considered for some  more particular situations before.

\end{abstract}

 \vskip 2mm


 \vskip 2mm
 \noindent {\bf 2000AMS Mathematical Subject Classification:} 70F10,70F15.\\

\section{Introduction}
Central configuration plays a very important role in many problems, such as Newtonian N-body problems. It is highly concentrated by mathematicians \cite{albouy2012finiteness,hampton2006finiteness,doi:10.1137/S0036141093248414,moeckel1990central,Palmore1975,saari1980role,smale1998mathematical,wintner1941analytical,Yu20122106}. To find  concrete central
configurations is very difficult, therefore we consider only some special situation, i.e., the central configurations formed by two twisted regular polygons, which are an extension of the results in \cite{doi:10.1137/S0036141093248414,Yu20122106}. The  motivation of this paper comes mainly from the results of \cite{doi:10.1137/S0036141093248414,Zhangperiodic,Yu20122106}.

\begin{definition}
A configuration $q = (q_1,q_2,\ldots,q_n)\in X \setminus \Delta$ is called a central configuration, if there exists a constant $\lambda \in \mathbb{R}$ such that
\begin{equation}\label{de}
\sum_{j = 1, j\neq k}^{n} \frac{m_j m_k}{|q_j - q_k|^3}(q_j - q_k) = -\lambda m_k q_k ,\quad 1\leqslant k \leqslant n.
\end{equation}
The value of $\lambda$ in (\ref{de}) is uniquely determined by
\begin{equation}\label{}
\lambda = \frac{U(q)}{I(q)},
\end{equation}
where
\begin{gather}\label{}
X = \{q= (q_1,q_2, \ldots, q_n) \in \mathbb{R}^{3n} | \sum_{i =1}^{n} m_i q_i = 0  \}, \\
\Delta = \{ q | q_j = q_k \,\, for\,some\,\, j\neq k  \},\\
U(q) = \sum_{1\leqslant j < k \leqslant n}{\frac{m_j m_k}{|q_j - q_k|}},\\
I(q) = \sum_{1\leqslant j \leqslant n} {m_j |q_j|^2}.
\end{gather}
\end{definition}
Consider the central configurations in $\mathbb{R}^3$ formed by one regular N-polygon and another regular L-polygon with distance
$h \geqslant 0 $(without loss of generality, we set $N\leqslant L$). It is assumed that the lower layer regular N-polygon lies in horizontal plane, and the upper
regular L-polygon parallels to the lower one and z-axis passes through both centers of two regular polygons. Suppose that the lower layer particles have masses $m_1,m_2,\ldots,m_N$
and the upper layer particles
have masses
 $\tilde{m}_1,\tilde{m}_2,\ldots,\tilde{m}_L$
respectively. For convenience, we may treat $\textbf{R}^3$ as the direct product of the complex plane and real axis when choosing the coordinates. Let $\rho_k$ be the $k-$th root of the N-roots of unity, $\xi_l$ be the $l-$th root of the L-roots of unity, i.e.,
\begin{gather}\label{}
\rho_k = e^{i\theta_k},\\
\xi_l  = e^{i\varphi_l},
\end{gather}
and let
\begin{gather}\label{}
\tilde{\rho}_l = a \xi_l \cdot e^{i\theta},
\end{gather}
where $a > 0, i = \sqrt{-1}, \theta_k = \frac{2k\pi}{N} (k =1,2,\ldots, N),\varphi_l = \frac{2l\pi}{L} (l =1,2,\ldots, L), 0 \leqslant \theta \leqslant 2\pi$, $\theta$ is called twisted angle.\\
It is assumed that $m_1,m_2,\ldots,m_N$ locates at the vertex $q_k$ of the lower regular N-polygon; $\tilde{m}_1,\tilde{m}_2,\ldots,\tilde{m}_L$ locates at the vertex $\tilde{ q}_l$ of the upper regular L-polygon,
where
\begin{gather}\label{}
q_k = (\rho_k, 0 ),\quad\tilde{ q}_l = (\tilde{\rho}_l, h).
\end{gather}
The center of mass is
\begin{equation}\label{}
z_0 = \frac{\sum_j {m_j q_j} + \sum_l{\tilde{m}_l \tilde{q}_l }}{M},
\end{equation}
where
\begin{equation}\label{}
M = \sum_j {m_j} + \sum_l{\tilde{m}_l }.
\end{equation}
Let
\begin{gather}\label{}
P_k = q_k - z_0, \quad P = (P_1,P_2,\ldots, P_N), \\
\tilde{P}_l = \tilde{q}_l - z_0, \quad \tilde{P} = (\tilde{P}_1,\tilde{P}_2,\ldots, \tilde{P}_L).
\end{gather}
If $P_1,P_2\ldots, P_N; \,\tilde{P}_1,\tilde{P}_2\ldots, \tilde{P}_L$ form a central configuration, then $\exists \lambda \in \mathbb{R}^+$, such that
\begin{gather}\label{}
\sum_{j = 1, j\neq k  }^{N }\frac{m_j }{|P_k - P_j|^3} {(P_k  - P_j)} + \sum_{j = 1 }^{L}\frac{\tilde{m}_j }{|P_k - \tilde{P}_j|^3} (P_k - \tilde{P}_j ) = \lambda P_k, 1\leqslant k\leqslant N, \\
\sum_{j = 1, j\neq k  }^{L }\frac{\tilde{m}_j }{|\tilde{P}_l - \tilde{P}_j|^3} {(\tilde{P}_l  - \tilde{P}_j)} + \sum_{j = 1 }^{N}\frac{m_j }{|\tilde{P}_k - P_j|^3} (\tilde{P}_l - P_j ) = \lambda \tilde{P}_l, 1\leqslant l\leqslant L.
\end{gather}
In the following, we only consider the case of $m_1 = \cdots = m_N = m$ and $\tilde{m}_1 = \cdots = \tilde{m}_L = bm$. Then
\begin{equation}\label{}
z_0 = \sum_{j}{(m_jq_j + \tilde{m}_j \tilde{q}_j)/M} = (0,0,\frac{bLh}{N + bL}),
\end{equation}
and the necessary conditions and sufficient
conditions for $P_1,P_2\ldots, P_N; \,\tilde{P}_1,\tilde{P}_2\ldots, \tilde{P}_L$ forming a central configuration are
\begin{equation}\label{}
A + b\sum_{j =1}^{L}{ \frac{1- a\cos{(\frac{2\pi j }{L} - \frac{2\pi k }{N} + \theta ) } }{[1+ a^2 -2 a \cos{( \frac{2\pi j}{L } - \frac{2\pi k}{N} + \theta ) +h^2  ]^{\frac{3}{2}}  } }  }  = \mu , \quad (1 \leqslant k \leqslant N)
\end{equation}
\begin{equation}\label{}
bB + \sum_{j = 1}^{N} {\frac{ 1- a^{-1}\cos{(\frac{2\pi j}{N} - \frac{2 \pi k }{L } - \theta)}  }{[1 + a^2 - 2a \cos{(\frac{2\pi j }{N } - \frac{2\pi k }{L} - \theta  ) + h^2 }   ]^{\frac{3}{2}}} } = \mu, \quad (1 \leqslant k \leqslant L)
\end{equation}

\begin{equation}\label{equation3}
\sum_{j = 1}^{L } { \frac{ \sin{(\frac{2\pi j}{L } - \frac{2 \pi k}{N} + \theta  )}}{[1 + a^2 - 2a\cos{(\frac{2\pi j}{L} - \frac{2\pi k}{N } + \theta  )} + h^2   ]^{\frac{3}{2}}     }} = 0,\quad (1 \leqslant k \leqslant N)
\end{equation}

\begin{equation}\label{equation4}
\sum_{j = 1}^{N } { \frac{ \sin{(\frac{2\pi j}{N } - \frac{2 \pi k}{L } - \theta  )}}{[1 + a^2 - 2a\cos{(\frac{2\pi j}{N} - \frac{2\pi k}{L } - \theta  )} + h^2   ]^{\frac{3}{2}}     }} = 0,\quad (1 \leqslant k \leqslant L)
\end{equation}

\begin{equation}\label{}
h \sum_{j = 1}^{L } { \frac{1}{[1 + a^2 - 2a\cos{(\frac{2\pi j}{L} - \frac{2\pi k}{N } + \theta  )} + h^2   ]^{\frac{3}{2}}     }} = \frac{\mu L h}{N + bL},\quad (1 \leqslant k \leqslant N)
\end{equation}

\begin{equation}\label{}
h \sum_{j = 1}^{N } { \frac{1}{[1 + a^2 - 2a\cos{(\frac{2\pi j}{N} - \frac{2\pi k}{L } - \theta  )} + h^2   ]^{\frac{3}{2}}     }} = \frac{\mu N h}{N + bL},\quad (1 \leqslant k \leqslant L)
\end{equation}

where
\begin{gather}\label{}
A = \sum_{j =1}^{N-1}{\frac{1-\rho_j}{|1-\rho_j |^3}  } > 0,\quad B = \sum_{j =1}^{L-1}{\frac{1-\xi_j}{|1-\xi_j |^3}  } > 0,\quad\mu=\frac{\lambda}{m}.
\end{gather}
For the central configurations of this type , R. Moeckel and C. Simo \cite{doi:10.1137/S0036141093248414} proved the following results with condition $\theta=0$, $L=N$:
\begin{theorem}\label{Moeckel1}
(R. Moeckel and C.Simo). When $h=0,\theta=0$, for every mass ratio $b$, there are exactly two planar central configurations consisting of two nested regular N-polygon. For one of these, the ratio $a$ of the sizes of the two polygons is less than 1, and for the other it is greater than 1. However, for $N\geq 473$ there is a constant $b_0(N)<1$ such that for $b<b_0$ and $b>\frac{1}{b_0}$, the central configuration with the smaller masses on the inner polygon is a repeller.
\end{theorem}
\begin{theorem}\label{Moeckel2}
(R. Moeckel and C.Simo). When $h^2 >0,\theta=0$, if $N<473$, there is a unique pair of spatial central configurations of parallel regular N-polygon. If $N\geq 473$, here are no such central configurations for $b<b_0(N)$. At $b=b_0$ a unique pair bifurcates from the planar central configuration with the smaller masses on the inner polygon. This remains the unique pair of spatial central configurations until $b=\frac{1}{b_0}$, where a similar bifurcation occurs in reverse, so that for $b>\frac{1}{b_0}$, only the planar central configurations remain.
\end{theorem}
X. Yu and S. Q. Zhang \cite{Yu20122106} proved the following results with condition $L = N$:
\begin{theorem}
If the central configuration is formed by two twisted regular N-polygon $(N\geq2)$ with distance $h\geq0$, then only $\theta=0$ or $\theta=\pi/N$. Specifically, if $a=1$ and $h=0$, i.e., two nested regular N-polygon are on the same unit circle, then only $\theta=\pi/N$.
\end{theorem}
\begin{corollary}\label{corollary1}
For $N\geq2,h=0$, if $a=1$, then $b=1$ and $\theta=\pi/N$, i.e., there is exactly one central configuration formed by two nested regular N-polygon on the same unit circle, which is the regular 2N-polygon.
\end{corollary}
\begin{corollary}\label{corollary2}
The configuration formed by two twisted regular N-polygon $(N\geq2)$ with distance $h\geq0$ is a central configuration if and only if the parameters $a,b,h$ satisfy the following relationships:
i. When $h=0$ and $a\neq1$\\
\begin{equation}b\left[\sum_{1\leq j\leq N} \frac{1-a\cos(\theta_j)}{(1+a^2-2a\cos(\theta_j))^{3/2}}-\frac{A}{a^3}\right]=\sum_{1\leq j\leq N} \frac{1-a^{-1}\cos(\theta_j)}{(1+a^2-2a\cos(\theta_j))^{3/2}}-A
\end{equation}\\
or\\
\begin{equation}b\left[\sum_{1\leq j\leq N} \frac{1-a\cos(\theta_j+\frac{\pi}{N})}{(1+a^2-2a\cos(\theta_j+\frac{\pi}{N}))^{3/2}}-\frac{A}{a^3}\right]=\sum_{1\leq j\leq N} \frac{1-a^{-1}\cos(\theta_j+\frac{\pi}{N})}{(1+a^2-2a\cos(\theta_j+\frac{\pi}{N}))^{3/2}}-A .
\end{equation}\\
ii. When $h>0$\\
\begin{eqnarray}
\left\{
\begin{array}{c}
ba\sum_{1\leq j\leq N} \frac{\cos(\theta_j)}{(1+a^2-2a\cos(\theta_j)+h^2)^{3/2}}=A-\sum_{1\leq j\leq N} \frac{1}{(1+a^2-2a\cos(\theta_j)+h^2)^{3/2}}\\
ba\left(\frac{A}{a^3}-\sum_{1\leq j\leq N} \frac{1}{(1+a^2-2a\cos(\theta_j)+h^2)^{3/2}}\right)=\sum_{1\leq j\leq N} \frac{\cos(\theta_j)}{(1+a^2-2a\cos(\theta_j)+h^2)^{3/2}}
\end{array}
\right.
\end{eqnarray}\\
or\\
\begin{eqnarray}
\left\{
\begin{array}{c}
ba\sum_{1\leq j\leq N} \frac{\cos(\theta_j+\frac{\pi}{N})}{(1+a^2-2a\cos(\theta_j+\frac{\pi}{N})+h^2)^{3/2}}=A-\sum_{1\leq j\leq N} \frac{1}{(1+a^2-2a\cos(\theta_j+\frac{\pi}{N})+h^2)^{3/2}}\\
ba\left(\frac{A}{a^3}-\sum_{1\leq j\leq N} \frac{1}{(1+a^2-2a\cos(\theta_j+\frac{\pi}{N})+h^2)^{3/2}}\right)=\sum_{1\leq j\leq N} \frac{\cos(\theta_j+\frac{\pi}{N})}{(1+a^2-2a\cos(\theta_j+\frac{\pi}{N})+h^2)^{3/2}}.
\end{array}
\right.
\end{eqnarray}
\end{corollary}
\begin{corollary}\label{corollary3}
For $N\geq2,h>0,a=1$, if the configuration formed by two twisted regular N-polygon $(N\geq2)$ with distance $h\geq0$ is a central configuration, then $b=1,\theta=0~\mbox{or}~\pi/N$, and there exists a unique $h$ for each $\theta$. In other words, there are exactly two spatial central configurations formed by parallel regular N-polygon which have the same sizes.
\end{corollary}
In this paper, we will prove the following main result:
\begin{theorem}\label{MainResult}
If $P_1,P_2\ldots, P_N; \tilde{P}_1,\tilde{P}_2\ldots, \tilde{P}_L$ form a central configuration, then $N=L$, i.e., two stacked regular polygons  forming a symmetrical central configuration (considered by us) have the same shape.
\end{theorem}

Then the necessary conditions and sufficient conditions for the central configurations formed by two twisted regular polygons are the conditions given in Corollary \ref{corollary1} and \ref{corollary2}. Furthermore,  some particular cases are already completely known, specially, the case of $\theta=0$ in Theorem \ref{Moeckel1} and \ref{Moeckel2}.
\section{Some Lemmas}
\begin{lemma}
Let $a_j > 0, 1\leqslant j \leqslant k, A_1 \geqslant A_2 \geqslant,\ldots,\geqslant A_k\geqslant 0$, then
\begin{equation}
\lim_{n\rightarrow\infty}({\sum_{1\leqslant j\leqslant k} {a_j A_j^n } })^{\frac{1}{n}} = A_1.
\end{equation}
\end{lemma}

\begin{lemma}
Given $a_j>0,1\leq j\leq k,A_1>\cdots> A_k>0$. For the function $f(x)=(\sum_{1\leq j \leq k}a_jA^x_j)^{\frac{1}{x}}, x \in (0,\infty)$, we have  \begin{equation}\label{}
f'(x)=(-A_1\ln a_1)\frac{1}{x^2}+o(\frac{1}{x^2}),
\end{equation}
when $x \rightarrow \infty$.
\end{lemma}

\begin{proof}
$f'(x)=(\sum_{1\leq j \leq k}a_jA^x_j)^{\frac{1}{x}}[\frac{x\sum_{1\leq j \leq k}a_jA^x_j\ln A_j}{\sum_{1\leq j \leq k}a_jA^x_j}-\ln{\sum_{1\leq j \leq k}a_jA^x_j}]/x^2$

$= \frac{1}{x^2}(\sum_{1\leq j \leq k}a_jA^x_j)^{\frac{1}{x}}[\frac{\ln A_1 + \sum_{2\leq j \leq k}b_jB^x_j\ln A_j}{1+\sum_{2\leq j \leq k}b_jB^x_j}x- \ln(a_1A^x_1)-\ln{1+ \sum_{2\leq j \leq k}b_jB^x_j}]$.

where $b_j=\frac{a_j}{a_1}, B_j=\frac{A_j}{A_1} (2 \leq j \leq k)$.

Then $B^x_j \rightarrow 0, xB^x_j \rightarrow 0$, when $x \rightarrow \infty$.

So $f'(x)=\frac{1}{x^2}(\sum_{1\leq j \leq k}a_jA^x_j)^{\frac{1}{x}}[-\ln a_1 + o(1)]=(-A_1\ln a_1)\frac{1}{x^2}+o(\frac{1}{x^2})$.
\end{proof}

\begin{lemma}\label{lemmaerror}
Let \begin{equation}g(x,\alpha)=\sum_{1\leq j\leq N} \frac{sin(\theta_j+\theta)}{(1+a^2-2acos(\theta_j+\theta)+x)^{\alpha}},\nonumber\end{equation}
where $\theta\in\left(0,\frac{\pi}{N}\right)$, $a>0, \alpha>0,x\geq0$. Then $g(x,\alpha)>0$ in $\{x:x\geq0\}$ provided $\alpha$ is sufficiently large.
\end{lemma}

\begin{proof}
Set $t=\frac{2a}{1+a^2+x}$, then $t \in (0,\frac{2a}{1+a^2}]$ and we need only to prove that
\begin{equation}f(t,\alpha)=\sum_{1\leq j\leq N} \frac{sin(\theta_j+\theta)}{(1-tcos(\theta_j+\theta))^{\alpha}}\nonumber\end{equation} is positive in
$(0,\frac{2a}{1+a^2}]$ for sufficiently large $\alpha$.

Firstly, we have
\begin{eqnarray}
f(t,\alpha)& = &\sum_{1\leq j\leq N}sin(\theta_j+\theta)\sum_{m\geq 0}[\frac{\alpha(\alpha+1)\cdots(\alpha+m-1)}{m!}t^m\cos^m(\theta_j+\theta)]\nonumber\\\
& = &\sum_{m\geq 0}[\frac{\alpha(\alpha+1)\cdots(\alpha+m-1)}{m!}(\frac{t}{2})^m]\sum_{1\leq j\leq N}[sin(m+1)(\theta_j+\theta)+(m-1)sin(m-1)(\theta_j+\theta)+\cdots]\nonumber\\\
& = & \frac{\alpha(\alpha+1)\cdots(\alpha+N-2)}{(N-1)!}(\frac{t}{2})^{N-1}N sin{(N \theta)}+o(t^{N-1})\nonumber\
\end{eqnarray}
So there exists $\delta_\alpha$ for any given $\alpha$ such that $f(t,\alpha)$ is positive for $t\in (0,\delta_\alpha]$.

Let \begin{equation}h(t,\alpha)=\{\sum_{0\leq j\leq [\frac{N-1}{2}]} \frac{sin(\theta_j+\theta)}{(1-tcos(\theta_j+\theta))^{\alpha}}\}^{\frac{1}{\alpha}}-\{\sum_{1\leq j\leq [\frac{N}{2}]} \frac{sin(\theta_j-\theta)}{(1-tcos(\theta_j-\theta))^{\alpha}}\}^{\frac{1}{\alpha}},\nonumber\end{equation}

 then $h(t,\alpha)$ and $f(t,\alpha)$ have the same sign.

We have $\frac{\partial h(t,\alpha)}{\partial \alpha}= [-\frac{\ln (\sin \theta)}{1-t\cos\theta}+\frac{\ln (\sin (\frac{2\pi}{N}-\theta))}{1-t\cos(\frac{2\pi}{N}-\theta)}]\frac{1}{\alpha^2}+o(\frac{1}{\alpha^2})$,
since $\frac{\ln (\sin (\frac{2\pi}{N}-\theta))}{1-t\cos(\frac{2\pi}{N}-\theta)}-\frac{\ln (\sin \theta)}{1-t\cos\theta}=\frac{(1-t\cos\theta)[\ln (\sin (\frac{2\pi}{N}-\theta))-\ln (\sin \theta)]+t[\cos(\frac{2\pi}{N}-\theta)-\cos\theta]}{[1-t\cos(\frac{2\pi}{N}-\theta)](1-t\cos\theta)}>0$,
hence $h(t,\alpha)$ is increasing about $\alpha$  provided $\alpha$ is sufficiently large.

So there exists some positive number $\delta$ (independent of $\alpha$ ) such that $h(t,\alpha)$ is positive in
$(0,\delta]$ for sufficiently large $\alpha$.

Secondly, $h(t,\alpha) \rightarrow h(t)$ for any $t \in [\delta,\frac{2a}{1+a^2}]$, where $h(t)=\frac{1}{1-t\cos\theta}-\frac{1}{1-t\cos(\frac{2\pi}{N}-\theta)}$, and there exists some positive number $\epsilon$ such that $h(t)\geq \epsilon >0$ for any $t \in [\delta,\frac{2a}{1+a^2}]$.

Since $h(t,\alpha)$ is increasing about $\alpha$  provided $\alpha$ is sufficiently large, by the well known theorem (Dini), we know that
$h(t,\alpha) \rightrightarrows h(t)$ on the compact interval $[\delta,\frac{2a}{1+a^2}]$, thus $h(t,\alpha)$ is positive in
$[\delta,\frac{2a}{1+a^2}]$ for sufficiently large $\alpha$.

As a result
$g(x,\alpha)$ is positive in $\{x:x\geq0\}$ provided $\alpha$ is sufficiently large.
\end{proof}

{\bf Remark.} We say some words about the proof of the \textbf{Lemma \ref{lemmaerror}} here. We found there was a gap in our original proof of the important \textbf{Lemma 2.10} in \cite{Yu20122106}. Since we din't notice that the problem of the uniform convergence of $h(t,\alpha)$ and the compactness of the whole interval we considered. So here we give the new proof of the \textbf{Lemma \ref{lemmaerror}} to correct the error in  \textbf{Lemma 2.10} of \cite{Yu20122106}.

Then we have the following important proposition which is a corollary of \textbf{Lemma \ref{lemmaerror}}. The detailed proof can be found in \cite{Yu20122106}.

\begin{lemma}\label{lemmayong}
If\begin{equation}\label{}
\sum_{j = 1}^{N } { \frac{ \sin{(\theta_j + \theta)}}{[1 + a^2 - 2a\cos{(\theta_j + \theta)} + h^2   ]^{\frac{3}{2}}     }} = 0,
\end{equation}
then $\theta  = \frac{j \pi }{N} (mod 2\pi)$ for some $1\leqslant j \leqslant 2N.$
\end{lemma}

\section{The proof of main results}

By(\ref{equation3}) and \textbf{lemma \ref{lemmayong}}, for$\forall k\in \{1,2,\ldots,L\}$there exist corresponding $l_k$ such that $1\leqslant l_k \leqslant 2L$ and
\begin{equation}\label{equation5}
\theta - \frac{2\pi k}{N} = \frac{l_k \pi}{L } (mod 2\pi).
\end{equation}
Then there exists some $ j \in \{1,2,\ldots, 2L  \}$ such that $\frac{2\pi}{N} = \frac{j\pi }{L} (mod 2\pi)$, moreover we have $N|(2L)$. Similarly, by(\ref{equation4}), for$\forall \nu\in \{1,2,\ldots,N\}$ there exist corresponding $n_\nu$ such that $1\leqslant n_\nu \leqslant 2N$ and
\begin{equation}\label{equation6}
-\theta - \frac{2\pi \nu}{L} = \frac{n_\nu \pi}{N } (mod 2\pi).
\end{equation}
Then  there exists some $ j \in \{1,2,\ldots, 2N  \}$ such that $\frac{2\pi}{L} = \frac{j\pi }{N} (mod 2\pi)$, thus
 we also have $L|(2N)$. Let $N= 2^n\cdot N_1, L = 2^l \cdot L_1$, where $N_1, L_1$ are  odd, then we have $N_1  = L_1,\ $and $n\leqslant l \leqslant n+1 $, so $L=N$ or $L = 2N$.

 In the following, we will prove $L \neq 2N$. Otherwise if $L = 2N$, we have $\theta = \frac{(4k+l_k)\pi}{2N} (mod 2\pi)(1\leq k\leq 2N, 1\leqslant l_k \leqslant 4N)$ by (\ref{equation5}) and $\theta =- \frac{(\nu+n_\nu)\pi}{N} (mod 2\pi)(1\leq \nu\leq N, 1 \leqslant n_\nu \leqslant 2N)$ by (\ref{equation6}), hence $l_k = 2l'_k$$(1\leqslant l'_k \leqslant 2N)$ for any $k\in \{1,2,\ldots,2N\}$ and $2k+l'_k+\nu+n_\nu = 0 (mod 2N)$ for any $k\in \{1,2,\ldots,2N\},\nu \in \{1,2,\ldots,N\}$, furthermore, there are both even and odd number in the numbers $n_\nu (\nu \in \{1,2,\ldots,N\})$.

Then it's easy to know that
$P_1,P_2\ldots, P_N; \tilde{P}_1,\tilde{P}_2\ldots, \tilde{P}_L$ form a central configuration if and only if $a,b,h$ satisfy the following equations (in fact, by symmetry, one can get the same result):

\begin{equation}\label{}
A + b\sum_{j =1}^{2N}{ \frac{1- a\cos{(\frac{j\pi }{N} ) } }{[1+ a^2 -2 a \cos{( \frac{j \pi}{N }) +h^2  ]^{\frac{3}{2}}  } }  }  = \mu ,
\end{equation}
\begin{equation}\label{}
bB + \sum_{j = 1}^{N} {\frac{ 1- a^{-1}\cos{(\frac{2j - 1}{N} \pi)}  }{[1 + a^2 - 2a \cos{(\frac{2j -1 }{N }\pi) + h^2 }   ]^{\frac{3}{2}}} } = \mu,
\end{equation}
\begin{equation}\label{}
bB + \sum_{j = 1}^{N} {\frac{ 1- a^{-1}\cos{(\frac{2j}{N} \pi)}  }{[1 + a^2 - 2a \cos{(\frac{2j}{N }\pi) + h^2 }   ]^{\frac{3}{2}}} } = \mu,
\end{equation}

\begin{equation}\label{equation10}
h \sum_{j = 1}^{2N } { \frac{1}{[1 + a^2 - 2a\cos{(\frac{j\pi }{N} )} + h^2   ]^{\frac{3}{2}}     }} = \frac{\mu L h}{N + bL},
\end{equation}
\begin{equation}\label{equation11}
h \sum_{j = 1}^{N } { \frac{1}{[1 + a^2 - 2a\cos{(\frac{2j -1}{N}\pi)} + h^2   ]^{\frac{3}{2}}     }} = \frac{\mu N h}{N + bL},
\end{equation}
\begin{equation}\label{equation12}
h \sum_{j = 1}^{N } { \frac{1}{[1 + a^2 - 2a\cos{(\frac{2j}{N}\pi)} + h^2   ]^{\frac{3}{2}}     }} = \frac{\mu N h}{N + bL}.
\end{equation}

\begin{lemma}\label{lemma3}
Let
\begin{equation}
 f(x) = \sum_{j=1}^{N}{\frac{1}{[1+a^2 -2a\cos{(\frac{2j-1}{N}\pi)} +x]^{\frac{3}{2}}}} - \sum_{j=1}^{N}{\frac{1}{[1+a^2 -2a\cos{(\frac{2j}{N}\pi)} +x]^{\frac{3}{2}}}}\quad x\in[0,+\infty) ,
 \end{equation}
 then for any $a\in(0,\infty)$, we have $f(x)<0$ for any  $x\in[0,+\infty)$ except the unique singularity $a=1,x=0$.
\end{lemma}
\begin{proof}
Set $t=\frac{2a}{1+a^2+x}$, then $t \in (0,1)$ and we need only to prove that
\begin{equation}g(t)=\sum_{1\leq j\leq N} \frac{1}{(1-tcos(\theta_j-\frac{\pi}{N}))^{ \frac{3}{2}}}-\sum_{1\leq j\leq N} \frac{1}{(1-tcos\theta_j)^{ \frac{3}{2}}}\nonumber\end{equation} is  negative in
$(0,1)$.

In fact, we have
\begin{eqnarray}
g(t)& = &\sum_{1\leq j\leq N}\sum_{m\geq 0}[\frac{(\frac{3}{2})(\frac{3}{2}+1)\cdots(\frac{3}{2}+m-1)}{m!}t^m(\cos^m(\theta_j-\frac{\pi}{N})-\cos^m\theta_j)]\nonumber\\\
& = &\sum_{m\geq 0}[\frac{(\frac{3}{2})(\frac{3}{2}+1)\cdots(\frac{3}{2}+m-1)}{m!}(\frac{t^m}{2^{m-1}})]\sum_{1\leq j\leq N}\{[\cos m(\theta_j-\frac{\pi}{N})+m \cos (m-2)(\theta_j-\frac{\pi}{N})+\cdots]\nonumber\\\
&-& [\cos m(\theta_j)+m \cos (m-2)(\theta_j)+\cdots]\}\nonumber\
\end{eqnarray}
Let $a_m=[\cos m(\theta_j-\frac{\pi}{N})+m \cos (m-2)(\theta_j-\frac{\pi}{N})+\cdots]
- [\cos m(\theta_j)+m \cos (m-2)(\theta_j)+\cdots]$, then it's easy to know that $a_m\leq 0$ and $g(t)$ is  negative in
$(0,1)$.
\end{proof}
By (\ref{equation11}),(\ref{equation12}) and \textbf{Lemma\ref{lemma3}}, we know $h=0$, and then there must be $a \neq1$, otherwise there will be  collision.
So $P_1,P_2\ldots, P_N; \tilde{P}_1,\tilde{P}_2\ldots, \tilde{P}_L$ form a central configuration if and only if $a,b,h$ satisfy the following equations
\begin{equation}\label{equation7}
A + b\sum_{j =1}^{2N}{ \frac{1- a\cos{(\frac{j\pi }{N} ) } }{[1+ a^2 -2 a \cos{( \frac{j \pi}{N })  ]^{\frac{3}{2}}  } }  }  = \mu ,
\end{equation}
\begin{equation}\label{equation8}
bB + \sum_{j = 1}^{N} {\frac{ 1- a^{-1}\cos{(\frac{2j - 1}{N} \pi)}  }{[1 + a^2 - 2a \cos{(\frac{2j -1 }{N }\pi)  }   ]^{\frac{3}{2}}} } = \mu,
\end{equation}
\begin{equation}\label{equation9}
bB + \sum_{j = 1}^{N} {\frac{ 1- a^{-1}\cos{(\frac{2j}{N} \pi)}  }{[1 + a^2 - 2a \cos{(\frac{2j}{N }\pi)  }   ]^{\frac{3}{2}}} } = \mu,
\end{equation}
\begin{lemma}\label{lemma4}
Let
\begin{equation}
f(x) = \sum_{j = 1}^{N} \frac{1- x^{-1}\cos{(\frac{2j -1}{N } \pi) } }{[1 + x^2 - 2x\cos{(\frac{2j -1}{N }\pi ) }  ]^{\frac{3}{2}}} - \sum_{j = 1}^{N} \frac{1- x^{-1}\cos{(\frac{2j}{N } \pi) } }{[1 + x^2 - 2x\cos{(\frac{2j }{N }\pi  )} ]^{\frac{3}{2}}}, x \in(0,1)\bigcup(1,\infty).
\end{equation}
then  $f(x) < 0$ for $x \in (1,\infty)$ and $f(x) > 0$ for $x \in (0,1)$.
\end{lemma}
\begin{proof}
Let \begin{equation}
h(x) = \sum_{j = 1}^{N} \frac{1}{[1 + x^2 - 2x\cos{(\frac{2j -1}{N }\pi ) }  ]^{\frac{1}{2}}} - \sum_{j = 1}^{N} \frac{1 }{[1 + x^2 - 2x\cos{(\frac{2j }{N }\pi  )} ]^{\frac{1}{2}}},\nonumber\
\end{equation}
then  \begin{equation}\label{equation15}
\frac{dh(x)}{dx} = -xf(x)
\end{equation}

Firstly, we will prove $f(x)$ is  negative in $(1,\infty)$.

Set $t(x)=\frac{2x}{1+x^2}$, then $t \in (0,1)$ and $\frac{dt(x)}{dx}=\frac{2(1-x^2)}{(1+x^2)^2}$,$\frac{dt(x)}{dx}>0$  for $x \in (0,1)$, $\frac{dt(x)}{dx}<0$  for $x \in (1,\infty)$. Furthermore, we have $h(x)=\frac{1}{(1+x^2)^{\frac{1}{2}}}g(t)$, where \begin{equation}g(t)=\sum_{1\leq j\leq N} \frac{1}{(1-tcos(\theta_j-\frac{\pi}{N}))^{ \frac{1}{2}}}-\sum_{1\leq j\leq N} \frac{1}{(1-tcos\theta_j)^{ \frac{1}{2}}}\nonumber\end{equation}

Since
\begin{eqnarray}
g(t)& = &\sum_{1\leq j\leq N}\sum_{m\geq 0}[\frac{(\frac{1}{2})(\frac{1}{2}+1)\cdots(\frac{1}{2}+m-1)}{m!}t^m(\cos^m(\theta_j-\frac{\pi}{N})-\cos^m\theta_j)]\nonumber\\\
& = &\sum_{m\geq 0}[\frac{(\frac{1}{2})(\frac{1}{2}+1)\cdots(\frac{1}{2}+m-1)}{m!}(\frac{t^m}{2^{m-1}})]\sum_{1\leq j\leq N}\{[\cos m(\theta_j-\frac{\pi}{N})+m \cos (m-2)(\theta_j-\frac{\pi}{N})+\cdots]\nonumber\\\
&-& [\cos m(\theta_j)+m \cos (m-2)(\theta_j)+\cdots]\}\nonumber\
\end{eqnarray}
Let $a_m=[\cos m(\theta_j-\frac{\pi}{N})+m \cos (m-2)(\theta_j-\frac{\pi}{N})+\cdots]
- [\cos m(\theta_j)+m \cos (m-2)(\theta_j)+\cdots]$, then it's easy to know that $a_m\leq 0$ and the first nonzero term is $a_N=-2$, furthermore, $g(t)$ and $\frac{dg(t)}{dt}$ is  negative in
$(0,1)$.

By (\ref{equation15}), we have
 \begin{equation}\label{equation16}
\frac{x(1+x^2)^{\frac{5}{2}}f(x)}{2} = (1+x^2)g(t)+\frac{(x^2-1)}{x}\frac{dg(t)}{dt},
\end{equation}
hence  $f(x) < 0$ for $x \in (1,\infty)$.

Secondly, let us prove $f(x)$ is positive in $(0,1)$.

In the following, we will prove $h(x)$ and $\frac{dh(x)}{dx}$ are both negative in $(0,1)$ by the method of R. Moeckel and C. Simo \cite{doi:10.1137/S0036141093248414}.

Let \begin{equation}
d(z) = \frac{1 }{(1 - z )^{\frac{1}{2}}}= \sum_{m\geq 0}d_mz^{m},\nonumber\
\end{equation}
then \begin{eqnarray}
h(x)& = &\sum_{j = 1}^{N} \frac{1 }{[1 - x\exp{(\frac{\sqrt{-1}(2j-1) }{N }\pi  )} ]^{\frac{1}{2}}}\frac{1 }{[1 - x\exp{(\frac{-\sqrt{-1}(2j-1) }{N }\pi  )} ]^{\frac{1}{2}}}\nonumber\\\
&-&\sum_{j = 1}^{N} \frac{1 }{[1 - x\exp{(\frac{\sqrt{-1}2j }{N }\pi  )} ]^{\frac{1}{2}}}\frac{1 }{[1 - x\exp{(\frac{-\sqrt{-1}2j }{N }\pi  )} ]^{\frac{1}{2}}}\nonumber\\\
&=& \sum_{j = 1}^{N}\sum_{m\geq 0}d_mx^{m}\exp{(\frac{\sqrt{-1} m (2j-1) }{N }\pi  )}\sum_{m\geq 0}d_mx^{m}\exp{(\frac{-\sqrt{-1} m (2j-1) }{N }\pi  )}\nonumber\\\
&-&\sum_{j = 1}^{N}\sum_{m\geq 0}d_mx^{m}\exp{(\frac{\sqrt{-1}2 m j }{N }\pi  )}\sum_{m\geq 0}d_mx^{m}\exp{(\frac{-\sqrt{-1}2 m j }{N }\pi  )}\nonumber\\\
&=& \sum_{j = 1}^{N}\sum_{m\geq 0}x^{m}\sum_{k+l=m}d_k d_l\exp{(\frac{\sqrt{-1} (k-l) (2j-1) }{N }\pi  )}\nonumber\\\
&-&\sum_{j = 1}^{N}\sum_{m\geq 0}x^{m}\sum_{k+l=m}d_k d_l\exp{(\frac{\sqrt{-1}2 (k-l) j }{N }\pi  )}\nonumber\\\
&=& N\sum_{m\geq 0}x^{m}\sum_{k+l=m,k\equiv l(mod N)}d_k d_l \cos (\frac{k-l}{N}\pi) - N\sum_{m\geq 0}x^{m}\sum_{k+l=m,k\equiv l(mod N)}d_k d_l\nonumber\\\
&:=& \sum_{m\geq 0}b_m x^{m}
\end{eqnarray}
It's easy to know that  all of the coefficients $b_m$ are nonpositive  and infinitely many are negative, then $h(x)$ and $\frac{dh(x)}{dx}$ are both negative in $(0,1)$. Hence  $f(x) > 0$ for $x \in (0,1)$ by (\ref{equation15}).
\end{proof}

Proof of Theorem \ref{MainResult}:
 \begin{proof}
 By(\ref{equation8}), (\ref{equation9}) and lemma(\ref{lemma4}), we know (\ref{equation7}), (\ref{equation8}) and (\ref{equation9})are unsolvable, i.e., if $L  = 2N$ , $P_1,P_2\ldots, P_N; \tilde{P}_1,\tilde{P}_2\ldots, \tilde{P}_L$   can not form a central configuration.
\end{proof}



\end{document}